\documentclass{article}
\usepackage{graphicx} % Required for inserting images
\usepackage{xspace}
\usepackage{color}
\usepackage[algo2e,ruled,vlined]{algorithm2e}
\usepackage{amsfonts}
\usepackage{amsthm}
\usepackage{amsmath}

\newtheorem{theorem}{Theorem}[section]
\newtheorem{claim}[theorem]{Claim}
\newtheorem{lemma}[theorem]{Lemma}

\newtheorem{corollary}[theorem]{Corollary}

\newcommand{\monotonesweep} {{\sf monotone-sweep}}
\newcommand{\treealgo} {{\sf tree-algo}}
\newcommand{\algotwog} {{\sf additive-error-algo}}
\newcommand{\naivealgo} {{\sf naive-algo}}
\newcommand{\optspanningtreealgo} {{\sf opt-spanning-tree-algo}}
\newcommand{\splitsweep} {{\sf split-tree-algo}}

\newcommand{\minMSet} {{\sf minMset}}
\newcommand{\minSMSet} {{\sf minSMset}}

\newcommand{\myparagraph}[1] {{\vspace*{0.08in}\noindent{\bf #1}~}}

\title{Minimum Monotone Tree Decomposition of Density Functions Defined on Graphs}
\author{%
  Lucas Magee,%
  \thanks{Department of Computer Science and Engineering, University of California, San Diego, 
          lmagee@ucsd.edu}\,
  Yusu Wang%
  \thanks{Hal{\i}c{\i}o\u{g}lu Data Science Institute, University of California, San Diego,
          yusuwang@ucsd.edu}\
}
\begin{document}
\maketitle

\begin{abstract}

Monotone trees - trees with a function defined on their vertices that decreases the further away from a root node one travels, are a natural model for a process that weakens the further one gets from its source.  Given an aggregation of monotone trees, one may wish to reconstruct the individual monotone components.  A natural representation of such an aggregation would be a graph.  While many methods have been developed for extracting hidden graph structure from datasets, which makes obtaining such an aggregation possible, decomposing such graphs into the original monotone trees is algorithmically challenging. 

Recently, a polynomial time algorithm has been developed to extract a minimum cardinality collection of monotone trees (M-Tree Set) from a given density tree - but no such algorithm exists for density graphs that may contain cycles. 
In this work, we prove that extracting such minimum M-Tree Sets of density graphs is NP-Complete.  
We additionally prove three additional variations of the problem - such as the minimum M-Tree Set such that the intersection between any two monotone trees is either empty or contractible (SM-Tree Set) - are also NP-Complete.
We conclude by providing some approximation algorithms, highlighted by a 3-approximation algorithm for computing the minimum SM-Tree Set for density cactus graphs.
\end{abstract}

\section{Introduction}
\label{sec:Introduction}

A common problem in modern data analysis is taking large, complex datasets and extracting simpler objects that capture the true nature and underlying structure. In this paper we are interested in the case when the input data is the aggregation of a collection of trees. In fact, each tree also has attributes over nodes (e.g., the strength of certain signal) which decreases monotonically from its root -- we call such a tree a monotone tree.  
Such trees come naturally in modeling a process that dissipates as it moves away from the root. One such example is in the construction of neuronal cells: a single neuron has tree morphology, with the cell body (soma) serving as the root. In (tracer-injection based) imaging of brains, the signal often tails off as it moves away from the cell body and out of the injection region, naturally giving rise to a rooted monotone tree.  Figure \ref{fig:monotone_tree_examples} (D) is centered around the soma of a single neuron within a full mouse brain imaging dataset with branches that get weaker as they get further from the soma.

Generally, we are interested in the following: given input data that is the aggregation of a collection of monotone trees, we aim to reconstruct the individual monotone trees. 
The specific version of the problem we consider in this paper is where the input data is a graph $G = (V, E)$ with a density function $f: V \to \mathbb{R}^{\geq 0}$ defined on its vertices. Our goal is to decompose $(G, f)$ into a collection of monotone trees $(T_1, f_1), \ldots, (T_k, f_k)$ whose union sums to the original $(G, f)$ at each $v \in V$. See Section \ref{sec:Preliminaries} for precise definitions. 
A primary motivation for considering graphs to be the input is because graphs are flexible and versatile, and recently, a range of methods have been proposed to extract the hidden graph structure from a wide variety of datasets; see e.g., \cite{Hastie84,Kegl02, Ozertem11, Aanjaneya11, 2011MNRAS,GSBW11,LRW14,CHS15,WWL15,DWW18,Magee22}.
In the aforementioned example of neurons, the discrete Morse-based algorithm of \cite{DWW18} has been applied successfully to extract a graph representing the summary of a collection of neurons \cite{Dingkang_skel, BMWL20}.  To extract the individual neurons from such a summary would be a significant achievement for the neuroscience community - which has developed many techniques to extract individual neuron skeletonizations from imaging datasets; see e.g., \cite{Hang_GTree2018, Quan_NeuroGPSTree_2016, smarttracing}.
However, going from a graph to a collection of trees poses algorithmic challenges.

The monotone-tree decomposition problem has been studied in the work of  \cite{baryshnikov2018minimal}, which develops a polynomial-time algorithm for computing the minimum cardinality set of monotone trees (M-Tree Set) of a density function defined on {\bf a tree} (instead of a graph).  However many applications for such a decomposition have graphs that may contain cycles, with the authors of \cite{baryshnikov2018minimal} explicitly mentioning a need for algorithms that can handle such input domains.

\myparagraph{New work.} We consider density functions defined on graphs, which we refer to as \emph{density graphs}. Our goal is to decompose an input density graph $(G, f)$ into as few monotone trees as possible, which we call the \emph{minimum M-Tree Decomposition problem}. See Section \ref{sec:Preliminaries} for formal definitions and problem setup. 
Unfortunately, while the minimum M-Tree Decomposition problem can be solved efficiently in polynomial time via an elegant greedy approach when the density graph is itself a tree \cite{baryshnikov2018minimal}, we show in Section \ref{sec:Complexity} that the problem for graphs in general is NP-Complete. In fact, no polynomial time constant factor approximation algorithm exists for this problem under reasonable assumptions (see Section \ref{sec:Complexity}). 
Additionally, we show NP-Completeness for several variations of the problem (Section \ref{sec:Complexity}). 
We therefore focus on developing approximation algorithms for this problem. In Section \ref{sec:Algorithms}, we first provide two natural approximation algorithms but with additive error. For the case of multiplicative error, we provide a polynomial time 3-approximation algorithm for computing the so called minimum SM-Tree Set of a density cactus graph.

\section{Preliminaries}
\label{sec:Preliminaries}

\subsection{Problem Definition}
We will now introduce definitions and notions in order to formally define what we wish to compute.
Given a graph $G(V, E)$, a \textit{density function} defined on $G$ is a function $f:V \rightarrow \mathbb{R}^{\geq 0}$.  
A \textit{density graph} $(G, f)$ is a graph $G$ paired with a density function $f$ defined on its vertices.
A \textit{monotone tree} is a density tree with a \textit{root} $v \in V$ such that the path from the root to every node $u \in V$ is non-increasing in density values.  See Figure \ref{fig:monotone_tree_examples} for explicit examples of density trees and monotone trees. While multiple nodes may have the global maximum value on the monotone tree, exactly one node is the root.  For example, in Figure \ref{fig:monotone_tree_examples} (B), either node with the global maximum value may be its root, but only one of them is the root.

\begin{figure*}
\begin{center}
\includegraphics[width = 0.8\linewidth]{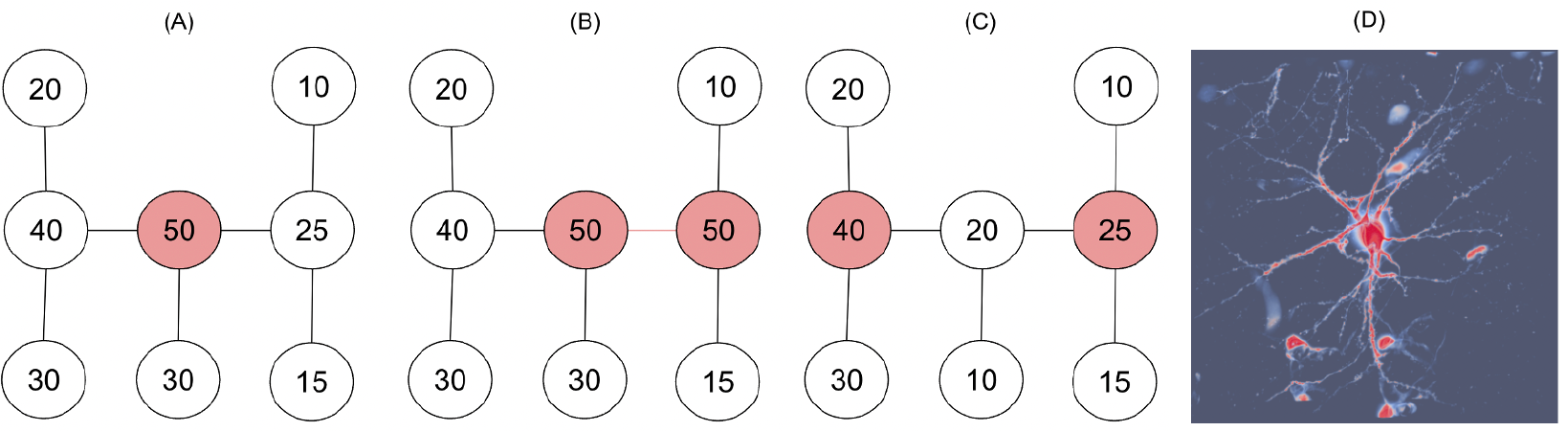}
\end{center}
\vspace*{-0.1in}
\caption{{\small (A) - (C) Contain examples of density trees with relative maxima colored red.  (A) shows a monotone tree. (B) shows a monotone tree with multiple nodes having the global maximum density value. (C) shows an example of a density tree that is not a monotone tree. (D) A zoom in an individual neuron within a full mouse brain imaging dataset.  The dataset is an fMOST imaging dataset that was created as part of the Brain Initiative Cell Census Network and is publicly available for download.
}}
\label{fig:monotone_tree_examples}
\end{figure*}

Given a density graph $(G(V, E), f)$, we wish to build a set of monotone subtrees $(T_{1}, f_{1}), (T_{2}, f_{2}), \dots, (T_{n}, f_{n})$ such that $T_i \subseteq G$ for all $i$ and $\sum_{i = 1}^{n}f_i(v) = f(v)$ for all $v \in V$. 
Note that if a node $v \in V$ is not in a tree $T_{i}$ then we say that $f_i(v) = 0$ and vice versa. 
We will refer to such a decomposition as a \emph{monotone tree (M-tree) decomposition} of the density graph, and refer to the set as an \textit{M-Tree Set} throughout the remainder of the paper.
An M-Tree Set is a \textit{minimum M-Tree Set} for a density graph if there does not exist an M-Tree Set of the density graph with smaller cardinality. 
An example of a density graph and a minimum M-Tree Set is shown in Figure \ref{fig:m-tree-set-example}. 
Note that a density graph may have many different minimum M-Tree Sets.
We abbreviate the cardinality of a minimum M-Tree Set for a density graph $(G, f)$ as $| \minMSet((G, f)) |$.

There are different types of M-Tree Sets that may be relevant for different applications. A \textit{complete M-Tree (CM-Tree) Set} is an M-Tree Set with the additional restriction that every edge in the density graph $G$ must be in at least one tree in the set.  
A \textit{strong M-Tree (SM-Tree) Set} is an M-Tree Set such that the intersection between any two trees in the set must be either empty or contractible.  We similarly abbreviate the cardinality of a minimum SM-Tree Set of $(G, f)$ as $| \minSMSet((G, f)) | $.
A \textit{full M-Tree (FM-Tree) Set} is an M-Tree Set such that for each element $(T_i(V_i,E_i), f_i)$, $f_{i}(v) = f(v)$ for the root node $v \in V_i$ of $(T_i, f_i)$.
The (minimum) M-Tree Set in Figure \ref{fig:m-tree-set-example} is also a (minimum) CM-Tree Set but is neither a SM-Tree Set nor a FM-Tree Set.

\begin{figure*}
\begin{center}
\includegraphics[width = 0.8\linewidth]{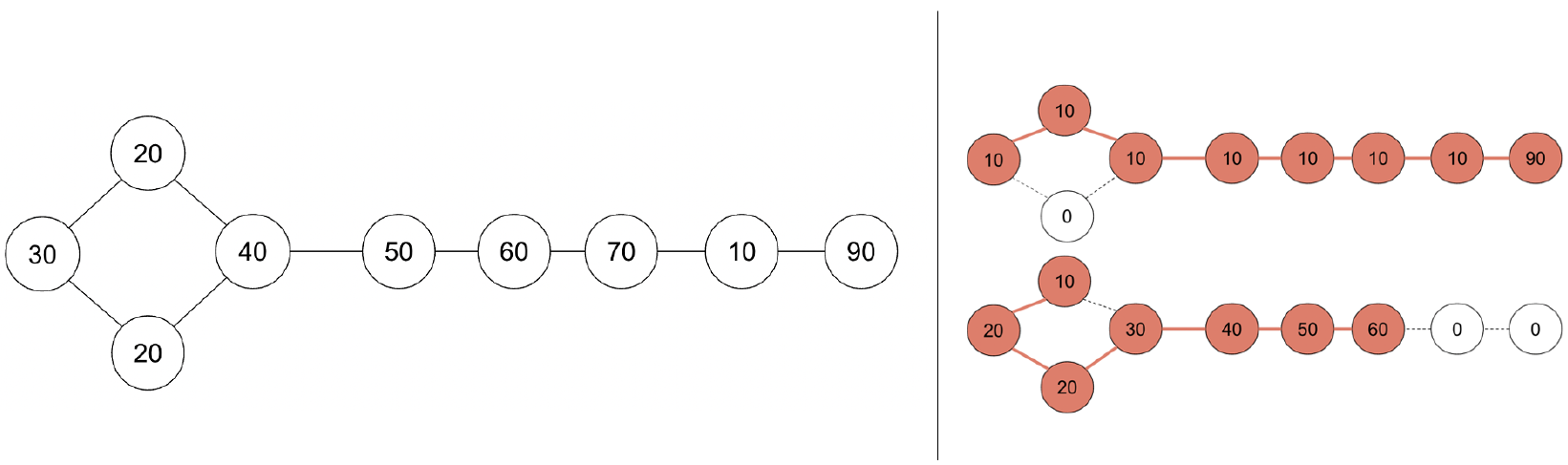}
\end{center}
\vspace*{-0.1in}
\caption{{\small A density graph (left) together with a minimum M-Tree Set (right).  Note that a minimum M-Tree Set is not necessarily unique for a density graph.
}}
\label{fig:m-tree-set-example}
\end{figure*}

\subsection{Greedy Algorithm for Density Trees \cite{baryshnikov2018minimal}}

We will now briefly describe the algorithm for computing minimum M-Tree Sets for density trees developed in \cite{baryshnikov2018minimal}, as some of the ideas will be useful in our work.  Please refer to \cite{baryshnikov2018minimal} for more details.
The approach of \cite{baryshnikov2018minimal} relies on a so-called \emph{monotone sweeping operation} to build individual elements of a minimum M-Tree Sets of density trees. Algorithm \ref{alg:monotonesweep} explicitly defines a generalized version of this operation that we will need in a later proof.

\begin{algorithm2e}
\KwIn{A density tree $(T(V, E), f)$, a starting node $v \in V$, and a staring value $\alpha$ such that $0 < \alpha \leq f(v)$} 
\KwOut{ A monotone subtree $(T', h_{f, v, \alpha})$ and a remainder $(T, R_{v, \alpha}f)$}

    % (Step 1) Orient each edge in $T$ away from $v$

    (Step 1) Initialize output density subtree $T'$ to only contain the input vertex $v$, with corresponding density function 
    $h_{f, v, \alpha}(v) = \alpha$ 
    
    (Step 2) Perform DFS starting from $v$.  For each edge $(u \rightarrow w)$ traversed:
    $h_{f, v, \alpha}(w) = \begin{cases} 
      h_{f, v, \alpha}(u) & f(w) \geq f(u) \\
      max(0, h_{f, v, \alpha}(u) - (f(u) - f(w))) & otherwise \\
   \end{cases}$\\
   
    Return monotone tree $(T', h_{f, v, \alpha})$ and remainder density tree $(T, R_{v, \alpha}f)$.
    
 \caption{\monotonesweep{}($(T(V, E), f), v \in V, \alpha$) }
 \label{alg:monotonesweep}
\end{algorithm2e}

The operation takes a density tree $(T(V, E), f)$, a node $v \in V$, and a starting function value $\alpha$ such that $0 < \alpha \leq f(v)$ as input. A monotone subtree $(T', h_{f, v, \alpha})$ and the remainder density tree $(T, R_{v, \alpha}f)$ where $R_{v, \alpha}f(u) = f(u) - h_{f, v, \alpha}(u)$ for all $u \in V$ is returned.

Algorithm \ref{alg:treealgo}, which outputs a minimum M-Tree Set of density trees, performs the monotone sweeping operation iteratively from certain nodes, called the \textit{mode-forced} nodes of the density tree.
To compute these mode-forced nodes, one iteratively remove leaves from the tree if their parent has greater or equal density.   Such leaves are referred to as \textit{insignificant vertices}. 
Once it is no longer possible to remove any additional nodes, the leaves of the remaining graph are the mode-forced nodes of the original density graph.

\begin{algorithm2e}
\KwIn{A density tree $(T(V, E), f)$} 
\KwOut{ A minimum M-Tree set of $(T(V, E), f)$}

    (Step 1) Find a mode forced vertex $v \in V$
    
    (Step 2) Perform \monotonesweep{}($(T(V, E), f), v, f(v)$) to build a single element of a minimum M-Tree Set.
    
    (Step 3) Repeat Steps 1 and 2 on remainder $(T, R_{v, f(v)}f)$ until no density remains.
    % \yusu{Is the remainder defined?}
    
 \caption{\treealgo{}($(T(V, E), f)$) }
 \label{alg:treealgo}
\end{algorithm2e}

An example of a single iteration of the tree algorithm is shown in Figure \ref{fig:tree-algo-example}.
The running time complexity of Algorithm \ref{alg:treealgo} is $O(n * | \mathrm{minMset}((T, f))|)$ where $n$ is the number of nodes in $T$.  We note that all M-Tree Sets of a density tree are also SM-Tree Sets, so Algorithm \ref{alg:treealgo} also outputs a minimum SM-Tree Set.

\begin{figure*}
\begin{center}
\includegraphics[width = 0.8\linewidth]{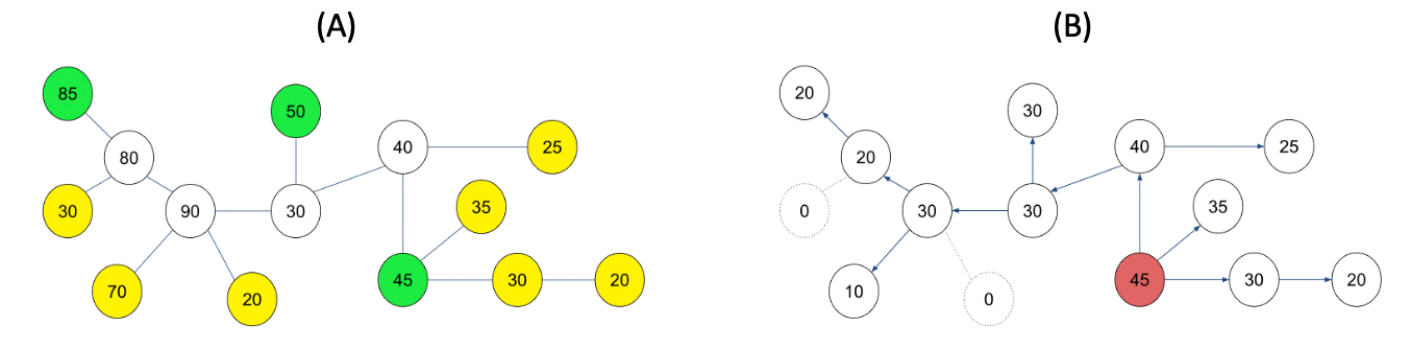}
\end{center}
\vspace*{-0.1in}
\caption{{\small (A) A density graph with mode-forced nodes colored green and insignificant vertices colored yellow.  (B)  A single element built by the monotone sweep operation from a mode forced node as performed in Algorithm \ref{alg:treealgo}. }}
\label{fig:tree-algo-example}
\end{figure*}

\subsection{Additional Property of Monotone Sweeping Operation}
\label{sec:PMSO}
Unfortunately, neither Algorithm \ref{alg:monotonesweep} nor Algorithm \ref{alg:treealgo} can be directly used to compute minimum M-Tree Sets of density graphs with cycles. 
Nevertheless, we can show Claim \ref{claim:sweep-as-much-as-possible} which will later be of use in developing approximation algorithms in Section \ref{sec:Algorithms}.

\begin{claim}\label{claim:sweep-as-much-as-possible}
Given a density tree $(T(V, E), f)$, let $v \in V$.  Let $a, b \in \mathbb{R}^{+}$ such that, without loss of generality,  $0< a < b \leq f(v)$.
Let $(T, R_{v, a}f)$ be the remainder of \monotonesweep{}($(T, f), v, a$).  We can define a similar remainder $(T, R_{v,b}f)$. Then we have 
$| \minMSet((T, R_{v,b}f)) | \leq | \minMSet((T, R_{v,a}f)) |$.
\end{claim}

\begin{proof}

We will prove the claim by contradiction.  Assume that $| \minMSet((T, R_{v,b}f)) | > | \minMSet((T, R_{v,a}f)) |$.

In particular, we will construct two new density trees, $(T_a, f_a)$ and $(T_b, f_b)$, as follows: 
$T_a$ is equal to our starting tree with the addition of two nodes $v_a$ and $v_{\infty}$, with two additional edges connecting to $v_a$ to both $v_{\infty}$ and $v$. Set $f_{a}(v_a) = a$ and  $f_{a}(v_{\infty}) = \infty$.  Similarly define $T_b$ and $f_b$.

Now imagine we run Algorithm \ref{alg:treealgo} on $(T_a, f_a)$. 
$v_{\infty}$ is a mode-forced node, and thus we can perform the first iteraiton in Algorithm \ref{alg:treealgo} on $v_{\infty}$. Sweeping from $v_{\infty}$ will leave remainder with a minimum M-Tree Set of size $| \minMSet((T_a, f_a)) |- 1$. The remainder is exactly the same as $(T, R_{v,a}f)$ at all nodes $v \in V$, and is zero at our newly added nodes. 
Hence, $| \minMSet((T_a, f_a)) | = | \minMSet((T, R_{v,a}f)) |  + 1$. 
Similarly, by performing Algorithm \ref{alg:treealgo} on $(T_b, f_b)$, we have $| \minMSet((T_b, f_b)) | = | \minMSet(T, R_{v,b}f)| + 1$

Now if our initial assumption is true, namely $|minMset((T, R_{v,b}f))| > |minMset((T, R_{v,a}f))|$, then by the above argument we have that 
\begin{align}\label{eqn:contradiction}
|\mathrm{minMSet}((T_b, f_b))| &> |\mathrm{minMSet}((T_a, f_a))|. 
\end{align} 
However, we could construct an M-tree set of $(T_b, f_b)$ as follows: First construct one monotone tree rooted at $v_\infty$ that leaves no remainder at both $v_\infty$ and $v_b$, then perform the monotone sweep operation starting at $v$ with starting value $a$ to build the rest of the component.
Note that the remainder after removing this tree is in fact $(T_a, R_{v,a}f)$, which we can then decompose using the minimum M-tree set of $(T_a, R_{v,a}f)$. 
In other words, we can find a M-tree set for $(T_b, f_b)$ with $| \minMSet(T, R_{v,a}f)| + 1 = |\minMSet(T_a, f_a)|$. 
This however contradicts with Eqn (\ref{eqn:contradiction}) (and the correctness of Algorithm \ref{alg:treealgo}). 
Hence our assumption cannot hold, and we must have that $| \minMSet(T, R_{v,b}f)| \leq |\minMSet(T, R_{v,a}f)|$. This proves the claim.

We note that while this proof is for M-Tree Sets specifically, the proof for SM-Tree Sets follows identical arguments.
\end{proof}

\section{Hardness Results}
\label{sec:Complexity}

Given that there exists a polynomial time algorithm for computing minimum M-Tree Sets of density trees, it is natural to ask whether or not such an algorithm exists for density graphs.  We prove Theorem \ref{thm:npcomp}, stating that the problem is NP-Complete. 

\begin{theorem}\label{thm:npcomp}
Given a density graph $(G(V, E), f)$ and a parameter $k$, determining whether or not there exists an M-Tree set of size $\leq k$ is NP-Complete.
\end{theorem}

\begin{proof}

It is easy to see that this problem is in NP, so we will now show it is also in NP-Hard.
First we consider a variation of the Set Cover problem where the intersection between any two sets is at most 1.  We refer to this problem as Set Cover Intersect 1 (SC-1).  SC-1 is a generalization of the NP-Complete problem of covering points in a plane with as few lines as possible \cite{MEGIDDO1982194}, and approximation bounds of SC-1 are well studied in \cite{kumar2000hardness}.
Given an instance of SC-1 ($m$ sets $S_1, S_2, \dots, S_m$ covering a universe of $n$ elements $e_1, e_2, \dots, e_n$, and a number $k$), we reduce to an instance of the M-Tree Set decision problem as follows:
\begin{itemize}
    \item Create a bipartite graph $G(V = A \cup B, E)$ equipped with a density function $f: V \rightarrow \mathbb{R}^{\geq 0}$ based on the input (SC-1) instance. 
    \item In particular, for each set $S_i$, add a node $a_{S_i}$ to $A$ and set $f(a_{S_i}) = | S_i |$.
    \item For each element $e_j$, add a node $b_{e_j}$ to $B$ and set $f(b_{e_j}) = 1$
    \item For each set $S_i$, add edge between $a_{S_i}$ and $b_{e_j}$ for each element $e_j \in S_i$.
\end{itemize}
An example of this reduction is illustrated in Figure \ref{fig:sc-1-reduction}.

\myparagraph{First Direction:
If there is a Set Cover of size $\le k$, then there is an M-Tree Set of density graph $(G, f)$ whose cardinality is $\leq k$.}

Let $S_{cover}$ be a set cover of size $n \leq k$. For each $S_i \in S_{cover}$, we will construct a monotone tree $(T_i, f_i)$ rooted at $a_{S_i}$.  In particular, $f_i(a_{S_i}) = f(a_{S_i})$.  Then, for each element $e_j \in S_i$, $T_i$ will include $b_{e_j}$ and the edge $(a_{S_i}, b_{e_j})$, with $f_i(b_{e_j}) = 1$. Note that if $e_j$ is an element in multiple sets in $S_{cover}$, simply pick one $S_i \in S_{cover}$ such that $e_j \in S_i$ to be the representative set of $e_j$. Finally, for each set $S_l \notin S_{cover}$, for each element $e_j \in S_l$, add the node $a_{S_l}$ and the edge $(b_{e_j}, a_{S_l})$ to $T_i$ with $f_i(a_{S_l}) = 1$, where $(T_i, f_i)$ is the monotone tree rooted at the node $a_{S_i}$ where $S_i \in S_{cover}$ is the representative set containing $e_j$.

Firstly, each element in the M-Tree Set is connected by construction.  The only nodes in an element $(T_i, f_i)$ are the root node $a_{S_i}$, where $S_i \in S_{cover}$, nodes of the form $b_{e_j}$, where $e_j \in S_i$, and nodes of the form $a_{S_l}$, where $S_l \notin S_{cover}$ and there exists $e_j$ in both $S_i$ and $S_l$.  Edges of the form $(a_{S_i}, b_{e_j})$ are part of the domain by construction and are included in $T_i$.  Similarly, edges of the form $(a_{S_l}, b_{e_j})$ are also part of the domain by construction and are included in $T_i$.  For each edge $(a_{S_l}, b_{e_j}) \in T_i$, there must also exist an edge $(a_{S_i}, b_{e_j})$.  Thus all nodes in $T_i$ are connected to $a_{S_i}$ - and in particular at most 2 edges away.

Secondly, each element in the M-Tree Set is a tree.  Consider element $(T_i, f_i)$.  By construction, if a cycle were to exist in $T_i$ it would have to be of the form $a_{S_i}, b_{e_p}, a_{S_l}, b_{e_q}, a_{S_i}$, where both $e_p$ and $e_q$ are in both $S_i$ and $S_l$.  However, such a cycle would imply that two sets have at least two elements in their intersection, which is not possible given we reduced from SC-1.

Next, each element in the M-Tree Set is a monotone tree. $f_i(v) = 1$ for all $v \in T_i$ that are not the root $a_{S_i}$ of $(T_i, f_i)$ and $f_i(a_{S_i}) \geq 1$.

Finally, $f(v) = \sum_{a = 1}^{n}f_i(v)$ for all $v \in G$. Each node $a_{S_i}$ such that $S_i \in S_{cover}$ is part of one monotone tree $(T_i, f_i)$ and $f_i(a_{S_i}) = f(a_{S_i})$.  Each node $b_{e_j} \in B$ is also part of only one monotone tree $(T_i, f_i)$ and $f_i(e_j) = 1 = f(e_j)$.  Finally, for a set $S_l \notin S_{cover}$, $a_{S_l}$ is included in $m = | S_l |$ monotone trees. For each such monotone tree $(T_i, f_i)$, $f_i(a_{S_l}) = 1$, thus $\sum_{a = 1}^{n}f_i(a_{S_l}) = m = f(a_{S_l})$.

Thus, we have proven that there exists a M-Tree Set of $(G, f)$ of size $\leq k$. 

\myparagraph{Second Direction: If there is an M-Tree Set of density graph $(G, f)$ of size $\leq k$, then there is a Set Cover of size $\leq k$.}

Let $\{(T_i, f_i)\}$ be an M-Tree set of density graph $(G, f)$ of size $k$. Each monotone tree $(T_i, f_i)$ in the set has a root node $m_i$.  If multiple vertices in $T_i$  have the maximum value of $f_i$ (as seen in Figure \ref{fig:monotone_tree_examples}(B)) simply set one of them to be $m_i$.  Each edge in $T_i$ has implicit direction oriented away from $m_i$. First we prove Lemma \ref{lemma:max-or-next-to-max}.
\begin{lemma}\label{lemma:max-or-next-to-max}
Let $b_{e_j} \in B$.  Either $b_{e_j}$ is the root of a monotone tree in the M-Tree Set or at least one of its neighbors is the root of a monotone tree in the M-Tree Set.
\end{lemma}

\begin{proof}
Assume $b_{e_j}$ is not a root of any monotone tree. Consider a monotone tree $(T_i, f_i)$ of the M-Tree Set containing $b_{e_j}$.  This means that $f_i(b_{e_j}) >0$. 
Consider node $a_{S_l}$ that is the parent of $b_{e_j}$ in $(T_i, f_i)$.
Assume $a_{S_l}$ is not the root node of $(T_i, f_i)$.
Because $a_{S_l}$ is not the root of the component, it must have a parent $b_{e_d}$.
Consider the remaining density graph $(G, g = f - f_i)$.
By definition of monotone tree, $0 < f_i(b_{e_j}) \leq f_i(a_{S_l}) \leq f_i(b_{e_d})$.
By construction, we also know $f(a_{S_l}) = \sum_{e_j \in S_l}f(b_{e_j})$.
Therefore, $g(a_{S_l}) > \sum_{e_j \in S_l}g(b_{e_j})$.
Because $a_{S_l}$ has more density than the sum of all of its neighbors in $(G, g)$, it is impossible for $a_{S_l}$ to not be the root of at least one monotone tree in any M-Tree Set of $(G, g)$.
Thus if $b_{e_j}$ is not the root of any monotone tree in the M-Tree Set, $a_{S_l}$ must be the root of a monotone tree in the M-Tree Set.
\end{proof}

We now construct a set cover from the M-Tree Set with the help of Lemma \ref{lemma:max-or-next-to-max}.  Initialize $S_{cover}$ to be an empty set.  For each $a_{S_i} \in A$ that is a root of a monotone tree in the M-Tree Set, add $S_i$ to $S_{cover}$.  Next for each $b_{e_j} \in B$ that is the root of a monotone tree in the M-Tree Set, if there is not already a set $S_i \in S_{cover}$ such that $e_j \in S_i$, choose a set $S_l$ such that $e_j \in S_l$ to add to the Set Cover.
Every element must now be covered by $S_{cover}$.  
A node $e_j$ that is not the root in any monotone tree in the M-Tree Set must have a neighbor $a_{S_l}$ that is a root in some monotone tree by Lemma \ref{lemma:max-or-next-to-max}. The corresponding set $S_l$ was added to $S_{cover}$ - thus $e_j$ is covered.
A node $e_m$ such that $b_{e_m}$ is the root of a monotone tree in the M-Tree Set must also be covered by $S_{cover}$ - as a set was added explicitly to cover $e_m$ if it was not already covered.
We've added at most one set to the cover for every monotone tree in the M-Tree Set, therefore $|S_{cover}| \leq k$.

Combining both directions, we prove that, given a SC-1 instance, we can construct a density graph $(G, f)$ such that there exists a set cover of size $\leq k$ if and only if the density graph has a M-tree Set of size $\leq k$. This proves the problem is NP-Hard, and thus the problem is NP-Complete.
\end{proof}

\begin{figure*}
\begin{center}
\includegraphics[width = 0.8\linewidth]{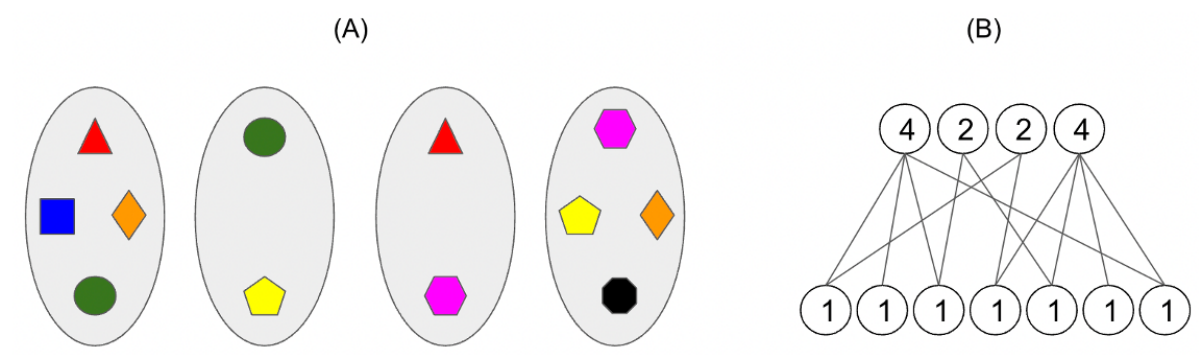}
\end{center}
\vspace*{-0.1in}
\caption{{\small (A) SC-1 instance with 4 sets and seven elements  (B) M-Tree decision problem instance created by following reduction outlined in proof of Theorem \ref{thm:npcomp}. The top row consists of nodes in $A \subset V$ in the bipartite graph, which are nodes representing sets, while the bottom row consists of nodes in $B \subset V$ in the bipartite graph, which are nodes representing elements.
}}
\label{fig:sc-1-reduction}
\end{figure*}

\subsection{Approximation Hardness}
From the proof of Theorem \ref{thm:npcomp}, it is easy to see that given an instance of SC-1, the size of its optimal set cover is equivalent to the cardinality of the minimum M-Tree Set of the density graph constructed in the reduction. 
Hence the hardness of approximation results for SC-1 translate to the minimum M-Tree Set problem too. 
We therefore obtain the following result, stated in Corollary \ref{cor:approx-bound} which easily follows from a similar result for SC-1. The SC-1 result from \cite{kumar2000hardness} is stated in Appendix \ref{app:sc1-bound}. We note that while Corollary \ref{cor:approx-bound} states the bound in terms of $n = $ of number of relative maxima, a similar bound can be obtained where $n = $ number of vertices.

\begin{corollary}\label{cor:approx-bound} 
There exists a constant $c > 0$ such that approximating the minimum M-Tree Decomposition problem within a factor of $c\frac{log(n)}{log(log(n))}$, where $n$ is the number of relative maxima on the given density graph, in deterministic polynomial time is possible only if $NP \subset DTIME(2^{n^{1 - \epsilon}})$ where $\epsilon$ is any positive constant less than $\frac{1}{2}$. 
\end{corollary}

\begin{proof}

Under the same assumptions mentioned above, there exists a $c > 0$ such that SC-1 cannot be approximated within a factor of $c\frac{log(n)}{log(log(n))}$, where $n$ is the number of elements in the universe \cite{kumar2000hardness}. 
We note that for a given SC-1 instance, performing the reduction to the M-Tree Set decision problem seen in the proof of Theorem \ref{thm:npcomp} results in a density graph with at most $\frac{n(n-1)}{2} + n$ relative maxima - the upper bound on the number of sets in the SC-1 instance.
Thus, the number of relative maxima on the density graph is $O(n^2)$.

For sufficiently large $n$, we have the following:

$c \frac{log(n^2)}{log(log(n^2))} = 2c \frac{log(n)}{log(2log(n))} = 2c \frac{log(n)}{log(log(n)) + 1} < 2c \frac{log(n)}{log(log(n))}$

Thus there exists a $c > 0$ such that minimum M-Tree Decomposition problem cannot be approximated within a factor of $c\frac{log(n^2)}{log(log(n^2))}$ under the same assumptions mentioned previously.  Because the number of relative maxima on the density graph is $O(n^2)$, we can substitute the number of relative maxima for $n^2$ to establish our final bound.

\end{proof}

\subsection{Variations of minimum M-Tree Sets are also NP-Complete}
In addition to proving that computing minimum M-Tree Sets of density graphs is NP-Complete, we have also proven Theorem \ref{thm:var-npcomp} in Appendix \ref{app:Complexity}.  The theorem states that computing the minimum CM-Tree Sets, minimum SM-Tree Sets, and minimum FM-Tree Sets of density graphs is also NP-Complete.

\begin{theorem}\label{thm:var-npcomp}
Given a density graph $(G(V, E), f)$ and a parameter $k$, determining whether or not there exists a CM-Tree Set, SM-Tree Set, or FM-Tree Set of size $\leq k$ are all NP-Complete.
\end{theorem}

It should be noted that Corollary \ref{cor:approx-bound} can be extended to CM-Tree Sets and FM-Tree Sets.  In contrast, SC-1 is not used in the NP-Complete proof for SM-Tree Sets.  Thus, Corollary \ref{cor:approx-bound} does not apply to minimum SM-Tree Set and there is hope we can develop tighter bounded approximation algorithms for this problem than for the other variations.

\section{Algorithms}
\label{sec:Algorithms}

\subsection{Additive Error Approximation Algorithms}
Now that we have shown that computing minimum M-Tree Sets of density graphs, as well as several additional variations, is NP-Complete, we focus on developing approximation algorithms.
We define two algorithms with different additive error terms. 
Firstly, we note that a naive upper bound for a given density graph is the number of relative maxima on the graph. We include Algorithm \ref{alg:naive} in Appendix \ref{app:Algorithms} to establish this naive upper bound.

Shifting focus to nontrivial approaches, Algorithm \ref{alg:approxalgo_2g} computes the minimum M-Tree Set of a density graph restricted to a spanning tree $T \subseteq G$.  We prove that $| \minMSet((T, f)) | \leq | \minMSet((G, f)) | + 2g$, where $g$ the \emph{genus} of $G$.  For a connected graph, $G(V, E)$, its genus is equal to $|E| - |V| + 1$, which is the number of independent cycles on the graph. This approximation error bound for Algorithm \ref{alg:approxalgo_2g} is stated in Theorem \ref{thm:opt_2g}.

\begin{algorithm2e}
\KwIn{A density graph $(G(V, E), f)$ such that $\beta_{1}G = g$} 
\KwOut{ An (S)M-Tree set of $G, f$ }

    (Step 1) Compute $g$ edges that if removed leave a spanning tree $T$ of $G$
    
    (Step 2) Compute minimum (S)M-Tree set of density tree $(T, f)$ via Algorithm \ref{alg:treealgo}
    
 \caption{\algotwog{}((G(V, E), f)) }
 \label{alg:approxalgo_2g}
\end{algorithm2e}

\begin{theorem}\label{thm:opt_2g}
Let $(G(V, E), f)$ be a density graph with $\beta_{1}G = g$. Let $k^*$ be the size of a minimum (S)M-Tree Set of $(G, f)$.  Algorithm \ref{alg:approxalgo_2g} outputs an (S)M-Tree Set of size at most $k^* + 2g$.
\end{theorem}

\begin{proof}

We need to prove Lemma \ref{lem:2g} to provide an upper bound on $| \minMSet(T, f) |$ for any spanning tree $T \subseteq G$.  
Algorithm \ref{alg:treealgo} will then output an M-Tree Set of size at most equal to the upper bound, thus completing our proof.  The proof is identical for SM-Tree Sets.

\begin{lemma}\label{lem:2g}
Let $(G(V, E), f)$ be a density graph with $\beta_{1}G = g$ and \\ $|\minMSet(G, f)| = k^{*}$. 
 For any spanning tree $T \subseteq G$, $|\minMSet(T, f)| \leq k^{*} + 2g$
\end{lemma}

Let $M = \{(T_i, f_i)\}$ be a minimum M-Tree set of $(G, f)$.
Let $E_{cut}$ be the set of $g$ edges that if removed from $G$ leave spanning tree $T$.
Firstly, we note that $|\minMSet(G, f)| \leq |\minMSet(T, f)|$, as any M-Tree Set of $(T, f)$ is also an M-Tree Set of $(G, f)$.

We will construct an M-Tree Set of $(T, f)$ from a minimum M-Tree Set of $G$.
For each monotone tree $(T_i, f_i) \in M$, consider an edge $e_{j} = (u, v) \in E_{cut}$ that is in $T_i$.  There is implicit direction to $e_{j}$ with respect to the root of $(T_i, f_i)$, meaning either (1) $(u \rightarrow v)$ or (2) $(v \rightarrow u)$.  
If (1) is the case, we can cut the branch rooted at $v$ off of $(T_i, f_i)$ to create two non-intersecting monotone trees.  See Figure \ref{fig:monotone-cut} for an example. 
We perform a similar operation if (2) is the case, but instead cut the branch rooted at $u$.  
Perform this cut for each edge in $E_{cut}$ to divide $(T_i, f_i)$ into, at most, $|E_{cut}| + 1$ non-intersecting monotone trees.
After dividing each tree into at most $|E_{cut}| + 1$ non-intersecting monotone trees, we make 2 key observations - (1) we still have a M-Tree Set of $(G, f)$ and (2) no edge in $E_{cut}$ is in any monotone tree in the M-Tree Set.  Thus the M-Tree Set is also an M-Tree Set of $(T, f)$.

We can shrink the size of this M-Tree Set by summing the components that share the same root.  In particular, consider an edge $e_{j} = (u, v) \in E_{cut}$. We have created as many as $k^{*}$ additional monotone trees rooted at $u$ and as many as $k^{*}$ additional monotone trees rooted at $v$.  Sum the monotone trees rooted at $u$ to create a single monotone tree rooted at $u$.  
The sum would clearly still be a monotone tree because all monotone trees are subtrees of tree $T$, so no cycle or non-non-increasing path from $u$ will be created. 
We can similarly do the same for $v$, and for all edges in $E_{cut}$.  This we have a new M-Tree Set of $(T, f)$, with (at most) an additional monotone tree rooted at each node of each edge in $E_{cut}$ when compared to the original M-Tree Set of $G$.  Thus $|\minMSet(T, f)|$ is bounded above by $k^{*} + 2g$.
\end{proof}

\begin{figure*}
\begin{center}
\includegraphics[width = 0.8\linewidth]{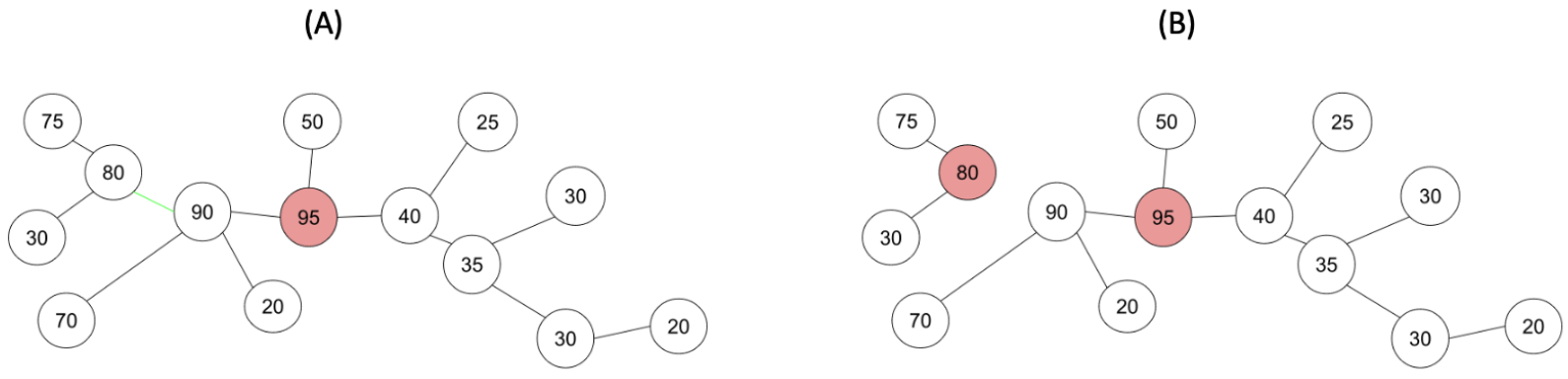}
\end{center}
\vspace*{-0.1in}
\caption{{\small (A) shows a single monotone tree with its root colored red and an edge colored green.  Cutting the green edge leaves us with two non-intersecting monotone trees shown in (B).
}}
\label{fig:monotone-cut}
\end{figure*}

\subsection{Approximation Algorithm for Minimum SM-Tree Sets of Density Cactus Graphs.}

A cactus graph is a graph such that no edge is part of more than one simple cycle \cite{harary1953number}.  See Figure \ref{fig:cactus-example} (A) for an example. 
Many problems that are NP-hard on graphs belong to P when restricted to cacti - such as vertex cover and independent set \cite{HARE1987437}. While we do not yet know whether or not computing a minimum M-Tree Set (or any variations) of density cactus graphs is NP-hard, we have developed a 3-approximation algorithm for computing the minimum SM-Tree Set of a density cactus graph. 

\begin{figure*}
\begin{center}
\includegraphics[width = 0.8\linewidth]{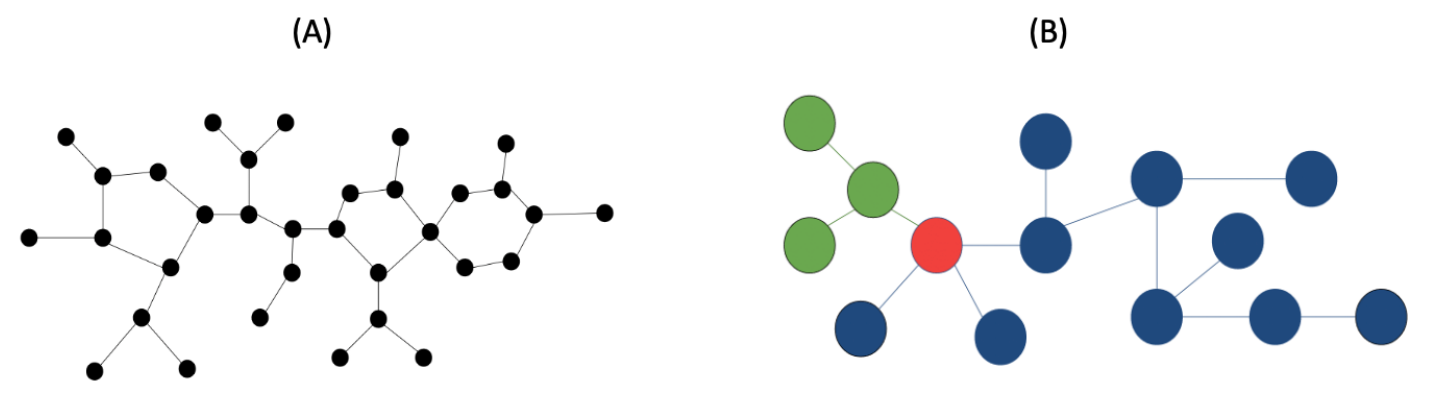}
\end{center}
\vspace*{-0.1in}
\caption{{\small (A) An example of a cactus graph - which is a graph such that no edge is part of more than a single simple cycle. It is essentially a tree of cycles graph - which is a graph such that no vertex is part of more than one simple cycle - with the exception that two simple cycles may share a single vertex.  Tree of cycles graphs are cactus graphs but cactus graphs (such as this one) are not necessarily tree of cycles graphs. (B) An example of an input for Algorithm \ref{alg:split-sweep}.  The density tree is broken into two subtrees (green and blue) that have a single node as intersection (red).  Monotone sweeping is performed iteratively at mode-forced nodes only in one of the subtrees.  Once the only remaining mode-forced nodes lie on the other tree, the output tuple containing the number of monotone sweeps performed and the remaining density at the intersection node is returned.
}}
\label{fig:cactus-example}
\end{figure*}

We first prove Theorem \ref{thm:sm-spanning-tree-3}, which states that for any density cactus graph $(G, f)$, there exists a spanning tree $T \subseteq G$ such that $| \minSMSet(T, f) |$ is at most 3 times $| \minSMSet(G, f) |$.

\begin{algorithm2e}
\KwIn{A density tree $(T(V, E), f)$ and two subtrees $T_1, T_2$ of $T$ that share a single node $v$ as intersection} 
\KwOut{ A tuple $(a, b)$ representing the number of monotone sweeps $a$ from mode-forced nodes on $T_1$ to make all mode-forced nodes on $T$ be part of $T_2$, and the remaining function value $b$ at $v$ after the monotone sweeps.}

    While there exists mode-forced node $u \in V$ off of $T_2$:
    
    - \monotonesweep{}($(T, f), u, f(u)$)
    
    Set $a = $ number of monotone sweeps performed
    
    Set $b = $ remaining density on $v$
    
    Return $(a, b)$
    
 \caption{\splitsweep{}((T(V, E), f), $T_1$, $T_2$) }
 \label{alg:split-sweep}
\end{algorithm2e}

\begin{theorem}\label{thm:sm-spanning-tree-3}
Let $(G(V, E), f)$ be a density cactus graph.  There exists a spanning tree $T$ of $G$ such that $| \minSMSet(T, f) | \leq 3 | \minSMSet(G, f) |$. 
\end{theorem}

\begin{proof} Let $M = \{(T_i, f_i)\}$ be a minimum SM-Tree Set of $(G, f)$, $k$ = $|M|$, and $\beta_{1}G = g$. 
Consider graph $G' = \bigcup_{i = 1}^{k}T_i$. Let $\beta_{1}G' = g'$. We note that $g' \leq g$.
and that $M$ is also a minimum SM-Tree Set of $G'$.
We will use $G'$ to help construct a spanning tree $T$ of $G$ with an SM-Tree Set with the desired cardinality.
Note that if $G'$ has no cycles then there obviously exists a spanning tree $T$ of $G$ such that $| \minSMSet(T, f) | = | \minSMSet(G, f) |$.  Additionally, if $g' = 1$, then creating $T$ by removing any edge from the simple cycle $| \minSMSet(T, f) | \leq | \minSMSet(G, f) | + 2$ (similar arguments to Lemma \ref{lem:2g} and Theorem \ref{thm:opt_2g}). Therefore assume $g' \geq 2$. Construct spanning tree $T$ as follows:
\begin{itemize}
    \item Add each edge in $G$ that is not part of a simple cycle.
    \item For each simple cycle in $G$ that is not in $G'$, add all edges of cycle to $T$ except for one missing in $G'$ (it does not matter which if multiple such edges exist).
    \item For each simple cycle in $G$ that is in $G'$, add all edges of that cycle to $T$ except for one (does not matter which).
\end{itemize}

Let $k'$ be the $| \minSMSet(T, f) |$. 
$k'$ is bounded above by $k + 2g'$, because removing edges that aren't used in the any monotone tree in $M$ from the domain will not change $| \minSMSet(G', f) |$.
Additionally, removing an edge from a simple cycle in $g'$ will increase the $| \minSMSet(G', f) |$ by at most 2 (again by Lemma \ref{lem:2g}).

$k'$ is also bounded below by $2 + g'$.  For each cycle in $G'$, the number of monotone trees in M that contain nodes in a simple cycle must be at least 3 - otherwise the set cannot be an SM-Tree Set.  So consider a leaf cycle $C_0$ in $G'$.  
We know that there are at least 3 monotones trees in $M$ that cover $C_0$ . For a cycle $C_1$ adjacent to $C_0$ in $G'$ that there is a single path between the two cycles, and the monotone trees that cover $C_0$ cannot completely cover $C_1$, otherwise $M$ would not be an SM-Tree Set. There must be at least one monotone tree with nodes on $C_1$ and no nodes on $C_0$.  
Continuing traversing the graph to all cycles and it is clear that for each cycle there must be an additional monotone tree added to the SM-Tree Set.  Thus, we cannot have an SM-Tree Set of size less than $2 + g'$.

From above, we have $\frac{k'}{k} \leq \frac{k + 2g'}{k} \leq \frac{2+g'+2g'}{2+g'} \leq \frac{3g' + 2}{g' + 2} < 3$.
\end{proof}

With Theorem \ref{thm:sm-spanning-tree-3} proven, we aim to compute the optimal density spanning tree of a density cactus graph.  
To help compute such an optimal density spanning tree, we first define Algorithm \ref{alg:split-sweep}.  
Given a density tree $(T(V, E), f)$ divided into two subtrees $T_1$ and $T_2$ that share a single node $v \in V$ as intersection, Algorithm \ref{alg:split-sweep} performs monotone sweeping operations on the mode-forced nodes 
 of $T_1$ until all mode-forced nodes of $(T, f)$ are on $T_2$. An example of a valid input is seen in Figure \ref{fig:cactus-example} (B).
 The output is a tuple $(a, b)$, where $a$ is the number of monotone sweeps performed and $b$ is the remaining density on $v$ after performing the monotone sweeps.  
 The tuple will essentially capture how helpful monotone sweeping from $T_1$ is for building a minimum (S)M-Tree Set on $T_2$. 
Algorithm \ref{alg:split-sweep} can be used to help compute the desired density spanning tree. In particular, it is used in Algorithm \ref{alg:cactus-3}, to cut the optimal edge from each cycle, one cycle at a time.  We prove Theorem \ref{thm:3-approx} which states that Algorithm \ref{alg:cactus-3} returns an SM-Tree Set at most 3 times larger than a minimum SM-Tree Set of a density cactus graph.
The running time of Algorithm \ref{alg:cactus-3} is $O(n^3)$ where $n$ is the number of nodes in the input cactus graph.  Algorithm 2 ($O(n^2)$) is performed once for each edge ($O(n)$) that is part of a simple cycle.

\begin{algorithm2e}
\KwIn{Density cactus graph $(G(V, E), f)$ } 
\KwOut{ SM-Tree Set of $(G, f)$ }

    If $G$ is a tree 
    
    - Compute optimal (S)M-Tree Set of $(G, f)$ using Algorithm \ref{alg:treealgo}.

    Else If $G$ has only a single cycle 
    
    - compute optimal sized (S)M-Tree Set of each density spanning tree of $G$ and return smallest cardinality (S)M-Tree Set.

    Else (G has multiple simple cycles)

    - Compute a leaf cycle $C = c_1, \dots, c_m$ connected to rest of cycles at $c_i$

    - Let $G_{C} =$ the simple cycle $C$ with all branches off of each node in the cycle - not including the branches off of $c_i$ that do not lead to other cycles in the graph. Let $G_{\bar{C}}=$ $T - G_{C} + c_i$. 

    - Fix a spanning tree $T_{G_{\bar{C}}}$ of $G_{\bar{C}}$.

    - For each spanning tree $T_{i}$ of $G_{C}$ computing \splitsweep{}($(G_{C} \cup G_{\bar{C}}, f), T_i, T_{G_{\bar{C}}}$)

    - Set $G = G(V, E - e^{*})$ such that $e^{*}$ is edge removed from $C$ that results in spanning tree with smallest output of \splitsweep{}.

    - Iterate until basecase (single cycle graph) is achieved 
    
 \caption{\optspanningtreealgo((G(V, E), f)) }
 \label{alg:cactus-3}
\end{algorithm2e}

\begin{theorem}\label{thm:3-approx}
Algorithm \ref{alg:cactus-3} outputs an SM-Tree Set that is at most 3 times the size of a minimum SM-Tree Set of the input density cactus graph.
\end{theorem}

\begin{proof}

Clearly, the algorithm outputs a minimum SM-Tree Set when the input domain is a tree. 
When $G$ contains a single cycle, by Lemma \ref{lem:2g} the outputted SM-Tree Set will have at most 2 more monotone trees than a minimum SM-Tree Set of $G$. 
Therefore, we only need to prove Theorem \ref{thm:3-approx} holds when $G$ contains multiple simple cycles.
Because $G$ is a cactus a leaf cycle $C = c_1, \dots, c_m$ exists.
Let $G_{C}$ be the graph of all nodes in the cycle and branches off of those nodes, excluding branches off of $c_i$ that do not lead to other cycles.
Let $G_{\bar{C}}$ be the graph of $G$ excluding $C$ and all branches off of $C$, except for the node $c_i$ itself.
$G_{C}$ has $m$ spanning trees, $T_{1}, \dots, T_{m}$ corresponding to the $m$ edges of $C$. Fix a spanning tree $T_{G_{\bar{C}}}$ of $G_{\bar{C}}$.
We next introduce Lemma \ref{lem:local-opt}.  

\begin{lemma}\label{lem:local-opt}
Let $T^{*} = $ spanning tree of $G_{C}$ such that the output of \splitsweep{}($T^{*} \cup T_{G_{\bar{C}}}, T^{*}, T_{G_{\bar{C}}}$) is minimized.  
$| \minSMSet((T^{*} \cup T_{G_{\bar{C}}}, f)) | \leq$ \\ $| \minSMSet((T_{k} \cup T_{G_{\bar{C}}}, f)) |$ for any spanning tree $T_k \subseteq G_C$. 
\end{lemma}

\begin{proof}
Lemma \ref{lem:local-opt} is proven by proving Claim \ref{claim:easy-opt-tree} and Claim \ref{claim:harder-opt-tree}.

\begin{claim}\label{claim:easy-opt-tree}
If \splitsweep{}($T_{j} \cup T_{G_{\bar{C}}}, T_{j}, T_{G_{\bar{C}}}$)[0] $<$ \splitsweep{}($T_{k} \cup T_{G_{\bar{C}}}, T_{k}, T_{G_{\bar{C}}}$)[0] then $| \minSMSet((T_{j} \cup G_{\bar{C}}, f)) | \leq | \minSMSet((T_{k} \cup G_{\bar{C}}, f)) |$.
\end{claim}

\begin{proof}
Let $T_j, T_k$ be spanning trees of $G_C$ such that
$a_j < a_k$, where $a_j =$ \splitsweep{}($T_{j} \cup T_{G_{\bar{C}}}, T_{j}, T_{G_{\bar{C}}}$)[0] and $a_k =$ \splitsweep{}($T_{k} \cup T_{G_{\bar{C}}}, T_{k}, T_{G_{\bar{C}}}$)[0].
Let $s^{*} = | \minSMSet((T_{G_{\bar{C}}}, f)) |$.

Algorithm \ref{alg:split-sweep} performs Algorithm \ref{alg:treealgo} sweeping from mode-forced nodes on $T_j$, but stops once mode-forced nodes only remain on $T_{G_{\bar{C}}}$. Thus it is still constructing minimum SM-Tree Sets but stopping short of completion.
The first element of the output of Algorithm \ref{alg:split-sweep} indicates the number of iterations required to have only mode-forced nodes on $T_{G_{\bar{C}}}$.  
$| \minSMSet((T_{j} \cup T_{G_{\bar{C}}}, f)) | \leq a_j + s^{*}$.
Similarly, $| \minSMSet((T_{k} \cup T_{G_{\bar{C}}}, f)) | \geq a_k + s^{*} - 1$.  These bounds prove the claim.

\end{proof}

\begin{claim}\label{claim:harder-opt-tree}
If \splitsweep{}($T_{j} \cup T_{G_{\bar{C}}}, T_{j}, T_{G_{\bar{C}}}$)[0] $=$ \splitsweep{}($T_{k} \cup T_{G_{\bar{C}}}, T_{k}, T_{G_{\bar{C}}}$)[0] and \splitsweep{}($T_{j} \cup T_{G_{\bar{C}}}, T_{j}, T_{G_{\bar{C}}}$)[1] $<$ \splitsweep{}($T_{k} \cup T_{G_{\bar{C}}}, T_{k}, T_{G_{\bar{C}}}$)[1] then $| \minMSet((T_{j} \cup G_{\bar{C}}, f)) | \leq | \minMSet((T_{k} \cup G_{\bar{C}}, f)) |$.
\end{claim}

\begin{proof}
Let $T_j, T_k$ be spanning trees of $G_C$ such that $a_j = a_k$ and $b_j < b_k$ where $(a_j, b_j) =$ \splitsweep{}($T_{j} \cup T_{G_{\bar{C}}}, T_{j}, T_{G_{\bar{C}}}$) and $(a_k, b_k) =$ \splitsweep{}($T_{k} \cup T_{G_{\bar{C}}}, T_{k}, T_{G_{\bar{C}}}$).

$a_j = a_k$ indicates that both $T_j$ and $T_k$ require the same number of iteration of monotone sweeps to leave mode-forced nodes on $T_{G_{\bar{C}}}$.  However, $b_j < b_k$ means that $T_j$ is more helpful than $T_k$ for reducing the minimum SM-Tree Set size on the remainder in the same number of monotone sweeps by Claim \ref{claim:sweep-as-much-as-possible}.  This proves the claim.
\end{proof}

This completes the proof of Lemma \ref{lem:local-opt}.

\end{proof}

It follows from Lemma \ref{lem:local-opt} that Algorithm \ref{alg:cactus-3} outputs a minimum SM-Tree Set on the density spanning tree that has the smallest sized minimum SM-Tree Set of all density spanning trees. Combining this with Theorem \ref{thm:sm-spanning-tree-3} finishes the proof.

\end{proof}

\section{Conclusion}
\label{sec:Conclusion}

Decomposing density graphs into a minimum M-Tree Set becomes NP-Complete when the input graph is not restricted to trees.  We proved that computing the minimum M-Tree, CM-Tree, SM-Tree, and FM-Tree Set of density graphs is NP-Complete. We provided additive error approximations algorithms for the minimum M-Tree Set problem, as well as developed a 3-approximation algorithm for minimum SM-Tree Sets for density cactus graphs.
Future work will be to close the gap between the bounds of approximation we have established with the error bounds of the algorithms we have developed.

\bibliographystyle{plain}
\bibliography{refs.bib}

\appendix

\section{Set Cover Intersection 1 Approximation Bound \cite{kumar2000hardness}}\label{app:sc1-bound}
\begin{theorem}\label{thm:sc1-bound} 
There exists a constant $c > 0$ such that approximating the SC-1 problem within a factor of $c\frac{log(n)}{log(log(n))}$, where $n$ is the number of elements in the universe, in deterministic polynomial time is possible only if $NP \subset DTIME(2^{n^{1 - \epsilon}})$ where $\epsilon$ is any positive constant less than $\frac{1}{2}$. 
\end{theorem}

\section{Complexity}
\label{app:Complexity}

In this section, we prove Theorem \ref{thm:var-npcomp}, which states many variations of the minimum M-Tree set problem are also NP-Complete.

\subsection{Proof of Theorem \ref{thm:var-npcomp}: CM-Tree Sets}
Firstly, the problem is clearly in NP.  We will follow the same reduction from SC-1 as seen in the proof of Theorem \ref{thm:npcomp} to prove NP-Hardness, with one additional step.

The first direction we follow identical arguments to create an M-Tree Set of appropriate size, but do not yet have a CM-Tree Set.
In particular, consider the M-Tree Set at the end of the proof - the only possible edges missing are edges $(a_{S_j}, b_{e})$ such that $S_j$ is in the Set Cover and contains $e$, but another set $S_i$ containing $e$ is in the Set Cover and $f_i(b_e) = 1$.
We will modify the M-Tree Set to ensure every such edge that is left out is included in a component.
Consider an element $e$ that is in $n$ sets in the set cover, where $n > 1$.
Let $(T_i, f_i)$ be the monotone tree in the M-Tree Set such that $f_i(b_e) = 1$.  
Set $f_i(b_e) = \frac{1}{n}$.
Additionally, for each set $S_j$ in the set cover such that $b_e \in S_j$, add $(a_{S_j}, b_e)$ to monotone tree $(T_j, f_j)$ and set $f_j(b_e) = \frac{1}{n}$.  Then, for each set $S_k$ such that $S_k$ is not in the set cover and $b_e \in S_k$, add $(b_e, a_{S_k})$ to each $(T_j, f_j)$ and set $f_j(a_{s_k}) = \frac{1}{n}$.
We still have an M-Tree Set, as each component clearly remains a monotone tree, and the sum of function values at each node is equal to what it was prior to the modification.
Once this modification is performed for every element contained within multiple sets in the set cover, we have an M-Tree set with every edge in the input domain included in at least one monotone tree.
The second direction is identical to the previous proof.

\subsection{Proof of Theorem \ref{thm:var-npcomp}: SM-Tree Sets}
Firstly, the problem is clearly in NP.
In order to prove this decision problem is NP-Hard - we first show that a specific instance of Vertex Cover - where for the given input graph $G(V, E)$, for any two verts $u, v \in V$, there is at most one vert $w \in V$ that is adjacent to both $u$ and $v$ - is NP-Complete.  This will limit the number of connected components in the intersection between two components to be at most one in our reduction to the M-Tree problem.

\begin{lemma}\label{thm:modified-vertex-cover-np-complete}
Given a graph $G(V, E)$ such that for any two verts $u, v \in V$, there is at most one vert $w \in V$ that is adjacent to both $u$ and $v$ and an integer $k$, determining whether or not there exists a vertex cover of size $\leq k$ is NP-Complete.
\end{lemma}

\begin{proof}

This is a specific instance of Vertex Cover and is clearly in NP.  To show it is in NP-Hard use the same reduction from 3-SAT to regular Vertex Cover as seen in \cite{garey1979computers}, but use a "restricted" version of 3-SAT where we can assume the following:
\begin{itemize}
    \item A clause has 3 unique literals
    \item A clause cannot have a literal and its negation
\end{itemize}
These assumptions are safe because we can transform any 3-SAT instance that has any such clauses to an equivalent 3-SAT instance with no such clauses in polynomial time.  Thus this restricted version of 3SAT is also NP-Complete.

Consider the graph created in the reduction of this restricted 3-SAT to Vertex Cover.  Two literal vertices will have no shared neighbors by design.  Any literal vertex and vertex in a clause will only have 1 neighbor if the literal vertex is also in the clause or the clause vertex is the negation of the literal vertex.  Two vertices in clauses will only have a single shared neighbor - the literal vertex if they are the same literal (and are thus in different clauses), or the final clause vertex if they are in the same clause. Thus, the reduction is also a reduction to our special Vertex Cover problem, and follows the exact same proof. 
\end{proof}

We will now reduce the special instance of Vertex Cover to the SM-Tree decision problem to prove NP-Hardness.
We follow a very similar reduction as seen in the proof of Theorem \ref{thm:npcomp}.  
We construct a bipartite graph $G(V = A \cup B, E)$, with nodes in $A$ corresponding to nodes in the instance of Vertex Cover, and nodes in $B$ representing edges in the instance of Vertex Cover.  
We then build density function $f$ on the nodes of $V$, setting $f(a_v) = $ degree of node $v$ in Vertex Cover instance for each node $a_v \in A$, and $f(b_e) = 1$ for each edge $e$ in the Vertex Cover instance.
Prove both directions the exact same way as shown in the proof of Theorem \ref{thm:npcomp}, but note that for the first direction, because of the restriction on our input graph, any two components of the decomposition can have at most a single vertex in their intersection.

\subsection{Proof of Theorem \ref{thm:var-npcomp}: FM-Tree Sets}
Firstly, the problem is clearly in NP.
To show the problem is in NP-Hard, we follow the exact same reduction from SC-1 as seen in proof of Theorem \ref{thm:npcomp}.  
For the first direction - we note that the M-Tree Set we have constructed is also an SM-Tree Set - as each set $S_i$ in the set cover is a root of a component $(T_i, f_i)$ such that $f_i(a_{S_i}) = f(a_{S_i})$.  
The second direction remains the same - though the argument that if a node $b_e$ is not the maximum of any monotone tree then one of its neighbors must be is slightly different.  
In this case, the neighbor must be a maximum in the same monotone tree it is a parent of $b_e$ in - not being so would contradict that the set is in fact a FM-Tree Set.

\section{Algorithms}\label{app:Algorithms}

\subsection{Naive Approximation Algorithm}

As stated in the main paper, given a density graph $(G, f)$, a natural upper bound for $| \minSMSet(G, f) |$ is the number of relative maxima on the density graph.  Algorithm \ref{alg:naive} constructs monotone trees rooted at each relative maxima on the input density graph.  
Starting at a root, depth-first search (DFS) is performed to reach every node that can be reached via a non-increasing path from the root.  
DFS stops once no nodes remain or all remaining nodes are not reachable from the root via a non-increasing path.  We call this graph traversal algorithm \textit{monotone DFS}.  
Perform monotone DFS from each relative maxima to build an M-Tree Set.  The M-Tree Set will have size at most one less than the number of relative maxima more than the size of a minimum M-Tree Set. 
Figure \ref{fig:naive-example} shows an example output of Algorithm \ref{alg:naive}.

\begin{algorithm2e}
\KwIn{A density graph $(G(V, E), f)$ } 
\KwOut{ An (S)M-Tree set of $(G(V, E), f)$ }

    (Step 1) Compute set $M$ containing the relative maxima of $f$ on $G$.
    
    (Step 2) For each relative maxima $m_i \in M$, perform monotone DFS to build a component $(T_i, f_i)$
    
    (Step 3) Return all $T_i, f_i$
    
 \caption{\naivealgo{}((G(V, E), f)) }
 \label{alg:naive}
\end{algorithm2e}

\begin{figure*}
\begin{center}
\includegraphics[width = 0.8\linewidth]{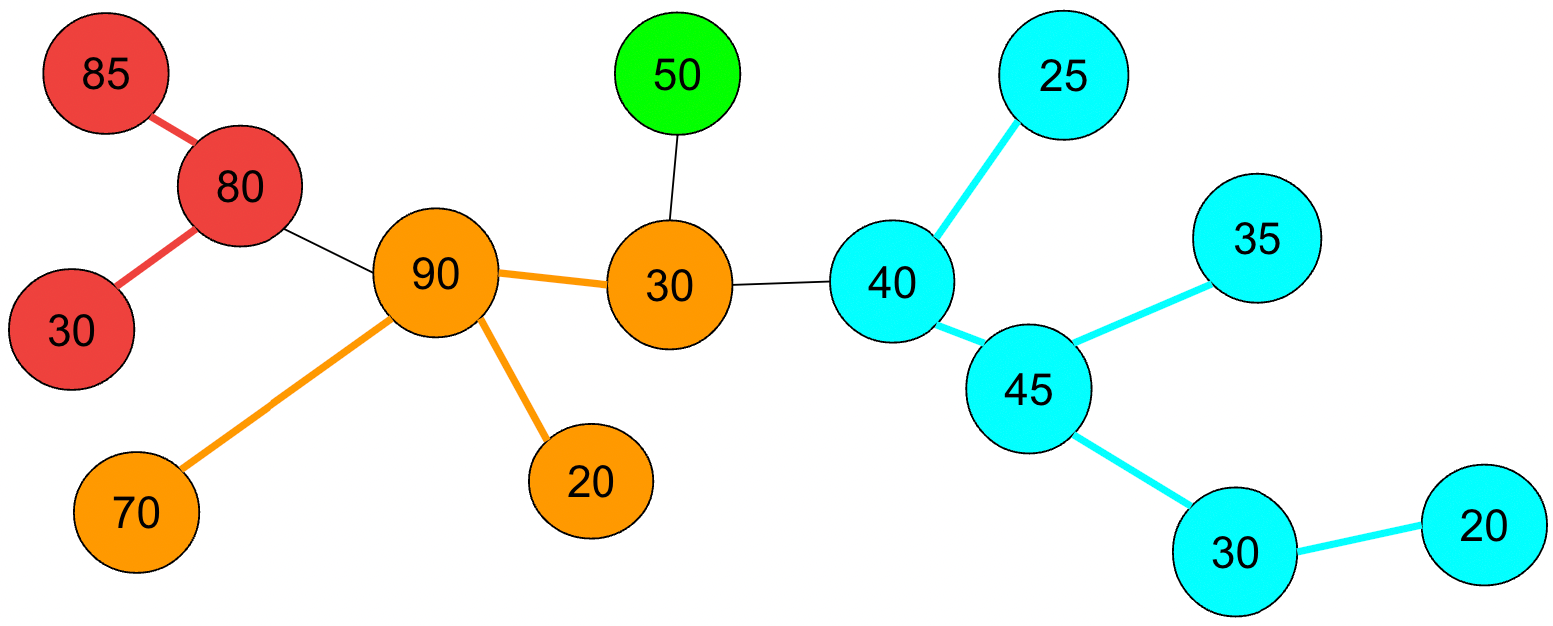}
\end{center}
\vspace*{-0.1in}
\caption{{\small An M-Tree Set with of a density tree with 4 monotone trees computed by Algorithm \ref{alg:naive}.
}}
\label{fig:naive-example}
\end{figure*}

\end{document}